\documentclass[11pt]{article}

\usepackage{epsf}
\usepackage[cp1250]{inputenc}
\usepackage{graphicx}
\usepackage{multirow}
\usepackage{graphics}
\usepackage{amsthm}
\usepackage{epsfig}
\usepackage{amsmath,amsfonts,amssymb,graphics,verbatim,epsfig
}

\newcommand{\cA}{\mathcal{A}}

\newcommand{\cF}{\mathcal{F}}

\newcommand{\cP}{\mathcal{P}}

\newcommand{\IE}{\mathbb{E}}

\newcommand{\IF}{\mathbb{F}}

\newcommand{\IP}{\mathbb{P}}
\newcommand{\IQ}{\mathbb{Q}}
\newcommand{\IR}{\mathbb{R}}
\newcommand{\IX}{\mathbb{X}}

\newcommand{\be}{\begin{eqnarray*}}
\newcommand{\ee}{\end{eqnarray*}}
\newcommand{\ben}{\begin{eqnarray}}
\newcommand{\een}{\end{eqnarray}}

\newtheorem{theorem}{Theorem}[section]

\theoremstyle{definition}

\newtheorem{problem}{Problem}[section]

\theoremstyle{remark}
\newtheorem{remark}[theorem]{Remark}

\numberwithin{equation}{section}

 \def\Q{{\mathbb Q}}

\def\ee{\varepsilon}

   \def\p{\pi}

\begin{document}

\title{Market risk modelling in Solvency II regime and hedging options not using underlying
}

\author{Przemys\l aw Klusik
\footnote{University of Wroc\l aw, pl.\ Grunwaldzki 2/4, 50-384
Wroc\l aw, Poland, E-mail: przemyslaw.klusik@math.uni.wroc.pl}
}
\maketitle

\begin{abstract}
In the paper we develop mathematical tools of quantile hedging in incomplete market. 
Those could be used for two significant applications:
\begin{enumerate}
\item
calculating the \textbf{optimal capital requirement imposed by Solvency II} (DIRECTIVE 2009/138/EC OF THE EUROPEAN PARLIAMENT AND OF THE COUNCIL) when the market and non-market risk is present in insurance company.

We show hot to  find the minimal capital $V_0$ to provide with the one-year hedging strategy for insurance company satisfying $E\left[{\mathbf 1}_{\{V_1 \geq D\}}\right]=0.995$, where $V_1$ denotes the value of insurance company in one year time and $D$ is the payoff of the contract.

\item finding a hedging strategy for derivative not using underlying but an asset with dynamics correlated or in some other way  dependent (no deterministically) on underlying. 

The work is a genaralisation of the work of  Klusik and Palmowski \cite{KluPal}.
\end{enumerate}
\medskip


\medskip

\end{abstract}
\section{Introduction}
DIRECTIVE 2009/138/EC OF THE EUROPEAN PARLIAMENT AND OF THE COUNCIL of 25 November 2009 on the taking-up and pursuit of the business of Insurance and Reinsurance (Solvency II) introduces new capital regimes on insurance companies. According to Section 4, Article 101, p. 3:
 \begin{quotation}
The Solvency Capital Requirement shall be calibrated so as to ensure that all quantifiable risks to which an insurance or reinsurance  undertaking  is  exposed  are  taken  into  account.  It  shall cover existing business, as well as the new business expected to be written over the following 12 months. With respect to existing business, it shall cover only unexpected losses. 

It shall correspond to the Value-at-Risk of the basic own funds of an insurance or reinsurance undertaking subject to a confidence level of 99,5 \% over a one-year period. 
 \end{quotation}

Further according to  p. 4:
 \begin{quotation}
The Solvency Capital Requirement shall cover at least the 
following risks: 

(a)  non-life underwriting risk;

(b)  life underwriting risk;

(c)  health underwriting risk;

(d)  market risk;

(e)  credit risk;

(f)  operational risk.

\end{quotation}

The question imposed by this regulation is how much money is enough to hedge the risk with the probability 0.995. What is important here from mathematical point of view, is that the risk here involves market and non-market factors, which means that it cannot be dealt using just real expectations probability measure. This is usually neglected by insurance companies although this neglectance oposses widely accepted Black-Scholes approach.

Mathematically speaking we ask for minimal $V_0$ ensuring the probability of satysfying all the claims $E\left[{\mathbf 1}_{\{V_1 \geq D\}}\right]\geq 0.995$, where $D$ denotes the contingent claim and $V_t$ denotes the value of hedging portfolio at time t. Equivalently we can fix the capital look for strategy wieth maximal probability of successful hedging

This problem was solved in literature only for complete markets (besides Sekine \cite{sekine} and  Klusik \& Palmowski \cite{KluPal}) , i.e for financial positions which don't allow for typical insurance risk.

Foellmer \& Leukert \cite{FL99} investigate the general semimartingale setting. Authors point out the optimal strategy for a complete market with maximal  $E\left[{\mathbf 1}_{\{V_T \geq D\}}\right]$.  The proofs are based on various versions of Neymann-Pearson lemma. Spivak \& Cvitanic \cite{spivak} study a complete market framework of assets modelled with Ito processes. They also constructed a strategy with maximal $E\left[{\mathbf 1}_{\{V_T \geq D\}}\right]$ but using different proof methods. They also implement this technique for market with partial observations. Finally they consider the case where the drift of the wealth process is a nonlinear (concave) function of the investment strategy of the agent.

Klusik, Palmowski, Zwierz \cite{KluPalZwie}  solve the problem of the quantile hedging from the point of view of a better informed agent acting on the market. The additional knowledge of the agent is modelled by a filtration initially enlarged by some random variable. 

Sekine \cite{sekine} considers a defaultable securities in very simple incomplete market, where security-holder can default at some random time and receives a payoff modelled by martingale process. Author shows strategy maximizing the probability of successful hedge.

The more complex incomplete market was considered by Klusik \& Palmowski \cite{KluPal}. They consider the equity-linked product where the insurance event can take a finite numer of states and is independent on financial asset modelled with geometric Brownian motion. They construct optimal strategy for both: maximal probability and maximal expected success ratio. In their framework the knowledge about the insurance event is not revealed before the maturity.

In this paper we state a general problem of optimizing probability of non-insolvency $E\left[{\mathbf 1}_{\{V_T \geq D\}}\right]$ in a incomplete market, as in  Klusik \& Palmowski \cite{KluPal}, but we allow very general flow of information outside the market and very general space of possible non-market events. As it was said at the beginning the solution of this problem gives a solution to Solvency II problem.

In fact the solution could be used not only for Solvency II, but also for pricing instruments in incomplete markets. This would include equity-linked or options on illiquid assets or traded only over-the-counter. The replicating strategy cannot be built in such case. It may suggest to price the option as an expectation in subjective probability measure or as the cost of superhedging (very costly!).  From practical point of view very often it may be  more apprioprate to hedge the claim with some dependent (for example correlated) liquid asset and smartly allow some risk, what actually to our knowledge is done by many (also of worldwide recognition) financial institutions without quantitive tools.

This paper is organized as follows. The section \ref{sec:mathematical_model}
introduces a model of financial market and the optimality problem. We also state and give a price of hedging and construct hedging strategy.

 In section \ref{sec:example} we provide the aplication of our result to hedging a European put option on nontradable asset. We calculate the cost of  hedging strategy with the other asset with price process partially dependent on the underlying. In numerical calculation we assume that both price processes are driven by correlated geometric Brownian motion.

\section{Mathematical model}\label{sec:mathematical_model}
Consider a discounted price process $\IX=(X_t)_{t \in [0,T]}$ which  is a semimartingale on a probability space $(\Omega, \cF, \IP)$ with filtration $\IF=(\cF_t)_{t \in [0,T]}$, $\cF_T=\cF$. Note that $\IF$ may be substantially greater than filtration generated by $\IX$. The interpretation is following: the knowledge modelled by $\IF$ could be augmented by information outside the market. 
The augmentation of filtration here could be interpreted as the information signal about non-market variables important to the value of contract. An example here could be the "life" part of information about the equity-linked contract. 
We will assume that $\cF=\cF_T$.

Denote the set of all equivalent martingale measures by $\cP$ and assume that the market does not allow for arbitrage, i.e $\cP \neq \emptyset$.

A self-financing admissible trading strategy is a pair  $(V_0,\xi)$ where  $V_0$ is constant and $\xi$ is a $\IF$-predictable process on
$[0,T]$ for which the value process $V_{t}:=V_{0}+\int_{0}^{t}\xi_u\;{\rm d}X_u$, $t\in[0,T],$ is well defined and $V_t\geq 0$ $\IP$-almost surely for all $t\in [0,T]$.

Fix an initial capital $\tilde{V}_0$ and denote by $\cA$ the set of all admissible strategies $(V_0,\xi)$ such that $V_0\leq \widetilde{V}_0$.

For nonnegative real $v,d$ define a \emph{success factor} $\phi^v_d$ assuming values in $[0,\infty]$ such that  $\phi^v_d$ is nondecreasing function of $v$ for all $d$.
The following functions can serve  as examples of success factor:  $\phi^v_d:=1_{\{v\geq d\}}$ and $\phi^v_d:={\mathbf 1}_{\{v \geq d)}+{\mathbf 1}_{\{v < d\}}\frac{v}{d}$.

For a contingent claim $D$ being a $\cF_T$-measurable  nonnegative random variable we formulate the following problem:
\begin{problem}\label{P1} Find $(V_0,\xi) \in \cA$  maximizing  \emph{expected success factor} $\IE^\IP\left[\phi^{V_T}_D\right]$.
\end{problem}
\begin{figure} \label{rys_pi}
	\centering
		\includegraphics[width=100mm]{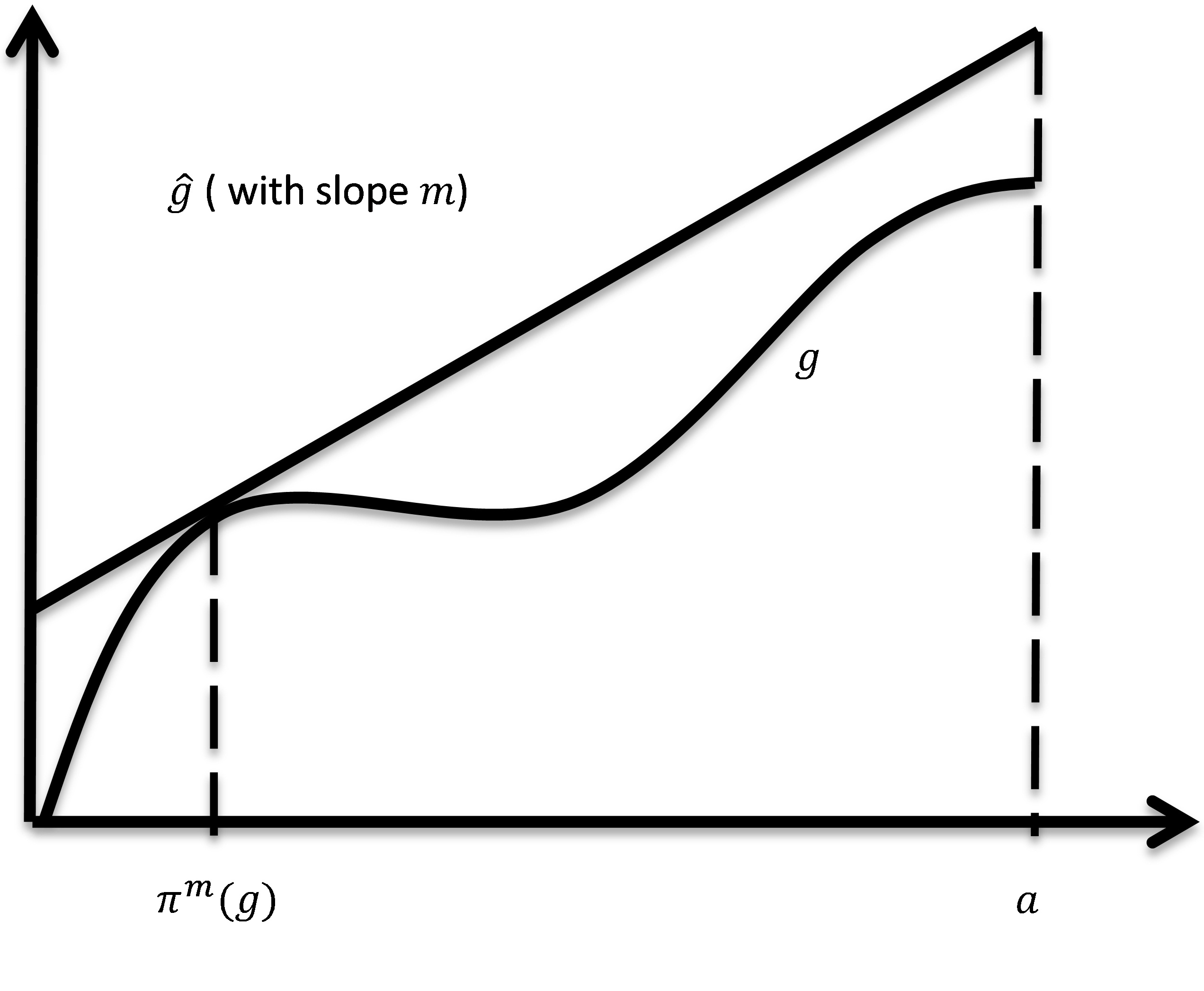} 
\caption{The picture shows the relation between $g,\hat{g},m$ and $\pi^m(g)$.}
\end{figure}
For any increasing function $g:[0,\infty]\rightarrow \IR$ and positive constant $m$ define 
\begin{equation}
\pi_m^g:=\min\{x:g(x)=\hat{g}(x)\}
\end{equation}
where $\hat{g}:[0,\infty]\rightarrow \IR$ denotes the highest line touching $g$ from above with the slope $m$ (see the figure).

Fix a measue $\IQ \in \cP$  and define 
\begin{equation}\label{defG}
G_\IQ(x):=\IE^{\IQ}\left[\frac{d\IP}{d\IQ}\phi^x_D\Big{|}\IX\right].
\end{equation}
Assume that there is a positive constant $m$  that $\pi^{G_\IQ}_m$ exists and is replicable with strategy $(V_0^*,
\xi^*)$ where $V_0^*=\widetilde{V}_0$, i.e.
\begin{equation}
\pi^{G_\IQ}_m=V^*_{0}+\int_{0}^{T}\xi^*_u\;{\rm d}X_u.
\end{equation}

\begin{remark}
Note that the assumption above is always fulfilled if all measures from $\cP$ coincide on $\sigma(\IX)$. This is true because $\pi^{G_\IQ}_m$ is  $\sigma(\IX)$-measurable. In particular this is the case for a complete market extended with contingent claims dependent on some randomness from outside the market( also in a situation when the information outside the market in revealed continously till the moment of maturity). In the next section we give a numerical procedure to find $m$ by Monte Carlo estimation.
\end{remark}
From now on we write $\pi$ as a shortcut to $\pi^{G_\Q}_m$.
\begin{theorem} \label{th_P1}
$(V_0^*,\xi^*)$ is a solution of Problem \ref{P1} with the expected success factor equal $\IE^\IP\left[\phi_D^{\pi}\right]$.
\end{theorem}
\begin{proof}  

 For any $(V_0,\xi) \in \cA$ holds $\IE^\IQ\left[V_T -\pi\right] \leq V_0- \widetilde{V}_0 \leq 0$.

For every $x$ and a. a. $\omega \in \Omega$ the inequality holds $\widehat{G}_\IQ^m(x)(\omega) \geq G_\IQ(x)(\omega)$ where $\widehat{G}_\IQ^m(x)=G_\IQ\left(\p\right)+m\left(x-\pi\right)$. Thus $\IE^\IQ\left[\widehat{G}_\IQ^m(V_T) \right] \geq\IE^\IQ\left[ G_\IQ(V_T)\right]$, i.e.
\begin{equation}
\IE^\IP\left[\phi^{V^*_T}_D\right]= \IE^\IP\left[\phi^{\pi}_D\right]\geq\IE^\IP\left[\phi^{\pi}_D\right]+m\IE^\IQ\left[V_T-\pi\right]\geq\IE^\IP\left[\phi^{V_T}_D\right].\nonumber
\end{equation}

Note that  $(V_0^*,\xi^*) \in \cA$ because $V^*_0=\widetilde{V}_0$, so left side of inequality is attainable.

\end{proof}
\section{Applications}\label{sec:example}
\subsection{Hedging contingent without underlying}

We consider a situation where we sell a put option on nontradable asset $Y$ with the payoff $D=(K-Y_T)^+$. We are going to hedge it using tradable asset $X$ with the strategy maximizing $\IP(V_T \geq D)$. Assume that the dynamics of two price processes is given with following equations
\begin{eqnarray}
dX_t&=&\mu_X X_t dt+\sigma_X X_tdW^X_t\nonumber, \quad X_0=x_0>0\\
dY_t&=&\mu_Y Y_t dt+\sigma_Y Y_tdW^Y_t\nonumber, \quad Y_0=y_0>0
\end{eqnarray}
where we assume a correlation $\rho$ between two Brownian motions $W^Y$ and $W^X$, i.e. $W^Y=\rho W^X+\sqrt{1-\rho^2}W$ where $W$ is some Brownian motion independent on $W^X$. We assume that the interest rate is equal zero. 
For $0<x<D$ we have
\begin{eqnarray}
& &G_\IQ(x)\nonumber\\
&=&\IE^{\IQ}\left[\frac{d\IP}{d\IQ}1_{\{D\leq x\}}\Big{|}\IX\right]= \frac{d\IP}{d\IQ}\IE^{\IQ}\left[1_{\{(K-Y_T)^+\leq x\}}\Big{|}\IX\right]\nonumber\\
&=&\frac{d\IP}{d\IQ}\IQ\left[K-x\leq y_0 \exp{\left\{ \mu_YT+\sigma_Y(\rho W_T^X+\sqrt{1-\rho^2}W_T)-\frac{1}{2}\sigma_Y^2T\right\}}\Big{|}\IX\right]\nonumber\\
&=&\frac{d\IP}{d\IQ}\IQ\left[\ln\left(\frac{K-x}{y_o}\right)\leq \mu_Y T+\sigma_Y(\rho W_T^X+\sqrt{1-\rho^2}W_T)-\frac{1}{2}\sigma_Y^2T\Big{|}\IX\right]\nonumber\\
&=&\frac{d\IP}{d\IQ}\IQ\left[\frac{\ln\left(\frac{K-x}{y_o}\right)- \mu_Y T-\sigma_Y\rho W_T^X+\frac{1}{2}\sigma_Y^2T}{\sigma_Y\sqrt{1-\rho^2}}\geq W_T\Big{|}\IX\right]\nonumber\\
&=& \exp{\left\{\frac{\mu_X}{\sigma_X}W_T^X+\frac{1}{2}\frac{\mu^2_X}{\sigma^2_X}T\right\}}\left(1-\Phi \left(\frac{\ln\left(\frac{K-x}{y_o}\right)- \mu_Y T-\sigma_Y\rho W_T^X+\frac{1}{2}\sigma_Y^2T}{\sigma_Y\sqrt{T(1-\rho^2)}}\right)\right)\nonumber,
\end{eqnarray}
where $\Phi$ denotes the cdf of standard normal distribution.
\begin{figure} \label{symulacja}
\centering
	\includegraphics[width=90mm]{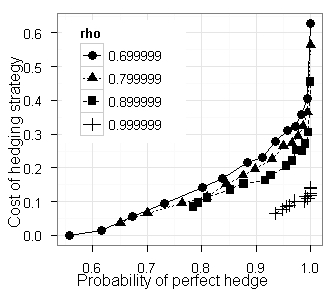} 
\caption{ The results of numerical simulations for described algorithm.}
\end{figure}
We describe the sketch of numerical algorithm basing on a Monte Carlo approach:
\begin{enumerate}
\item Fix a real number $m\geq 0$ and integers $N_x>0, N_W>0$.
\item Draw a sample $w_1,\ldots,w_{N_W}$ from normal distribution with mean 0 and variance $T$.
\item For $i=1,\ldots,N_W$ find  such $x$ from set $\{0,\frac{1K}{N_x},\frac{2K}{N_x},\ldots,\frac{N_xK}{N_x}\}$ maximizing expression 
\begin{equation} \label{dynamics}
 \exp{\left\{\frac{\mu_X}{\sigma_X}w_i+\frac{1}{2}\frac{\mu^2_X}{\sigma^2_X}T\right\}}\left(1-\Phi \left(\frac{\ln\left(\frac{K-x}{y_o}\right)- \mu_Y T-\sigma_Y\rho w_i+\frac{1}{2}\sigma_Y^2T}{\sigma_Y\sqrt{T(1-\rho^2)}}\right)\right)-mx\nonumber
\end{equation}
and denote it by $x^{\max}(i)$.
\item The solution is following: For an initial capital equal to
\begin{equation}
\frac{1}{N_W}\sum_{i=1}^{N_W} \exp{\left\{-\frac{\mu_X}{\sigma_X}w_i-\frac{1}{2}\frac{\mu^2_X}{\sigma^2_X}T\right\}}x^{\max}(i)\nonumber
\end{equation} 
maximal expected success factor is equal
\begin{equation}
\frac{1}{N_W}\sum_{i=1}^{N_W} \left(1-\Phi \left(\frac{\ln\left(\frac{K-x^{\max}(i)}{y_o}\right)- \mu_Y T-\sigma_Y\rho w_i+\frac{1}{2}\sigma_Y^2T}{\sigma_Y\sqrt{T(1-\rho^2)}}\right)\right).\nonumber
\end{equation} 
\end{enumerate}
Different $m$ would give a different  initial capital and expected success factor.

The figure shows the results of simulations for put option maturing at time $T = 1$ with the strike $K = 1$.
The price dynamics  follows \ref{dynamics} with parameters $\mu_X =\mu_Y= 0.1$, $\sigma_X = \sigma_Y= 0.3$, $Y_0 =X_0= 1$. The diagram illustrates the dependence for different levels of $\rho$. 

One could verify that as expected  if $X$ almost mimicks $Y$  (i.e. $\rho$ is almost 1),  the hedging strategy with $X$ should be very close to the hedging strategy with $Y$ if $Y$ was tradable. The cost of the last  is equal $0.119235$ as the result of standard Black-Scholes formula.

\subsection{Applications: Solvency II}
Usually during capital modelling insurance companies ignore the basic difference between market risks and insurance (and other nonmarket risks): the fact that market risks can be  hedged away using underlying asset in our framework.
In our seeting this would mean  taking $T=1$, $\phi^v_d:=1_{\{v\geq d\}}$ and $D$ that represent all  liabilities of insurance company at time $1$.

It is worth to stress that our solution give gives a strategy to get minimal probability of insolvency. Equivalently for given probability (here 0.995) we get \underline{minimal} capital needed. This solution doesn't assume static or almost static posioni, what is usually done in practic, but point out the best strategy(possibly dynamic).

\section{Acknowledgements}
This author's research was supported by the Ministry of Science and
Higher Education grant NCN 2011/01/B/HS4/00982.

\bibliographystyle{plain}
{\small\bibliography{PhD}}

\end{document}